\documentclass{amsart}

\usepackage[utf8]{inputenc}
\usepackage[english]{babel}
\usepackage{csquotes}
\usepackage{hyperref}

\usepackage{amssymb}
\usepackage{dcolumn}
\usepackage{enumerate}
\usepackage{mathdots}
\usepackage{tabularx}
\usepackage{subfigure}
\usepackage{amsmath}
\usepackage{amsthm}
\usepackage{tikz} 
\usepackage{float}
\usepackage{mathtools}
\usepackage{graphicx}
\usepackage{pgfplots}
\usepackage{cite}
\usepackage{pgf,tikz}
\usepackage{mathrsfs}
\pgfplotsset{compat=1.17} 

\numberwithin{equation}{section}

\newtheorem{thm}{Theorem}[section]

\newtheorem{lem}[thm]{Lemma}
\newtheorem{rem}[thm]{Remark}

\newtheorem{obs}[thm]{Observation}

\theoremstyle{definition}
\newtheorem{defn}{Definition}[section]

\newcommand{\C}{\mathcal{C}}
\newcommand{\Aut}{\text{Aut}}
\newcommand{\wt}[1]{\widetilde{#1}}
\newcommand{\smn}{S(m,n)}

\newcommand{\p}[2]{P(#1,#2)}
\newcommand{\up}[2]{\mathbb{P}_{#1}(#2)} 
\newcommand{\pmn}{\p{m}{n}}
\newcommand{\upmn}{\up{m}{n}}
\newcommand{\flop}{\overleftarrow{r}}

\begin{document}

\author{Sa\'ul A. Blanco and Charles Buehrle}
\title{Lengths of Cycles in Generalized Pancake Graphs}

\address{Department of Computer Science, Indiana University, Bloomington, IN 47408}
\email{sblancor@indiana.edu}
\address{ Department of Mathematics, Physics, and Computer Studies, Notre Dame of Maryland University, Baltimore, MD 21210}
\email{cbuehrle@ndm.edu}

\date{July 11, 2023}

\bibliographystyle{ieeetr}

\begin{abstract}
In this paper, we consider the lengths of cycles that can be embedded on the edges of the generalized pancake graphs which are the Cayley graph of the generalized symmetric group $S(m,n)$, generated by prefix reversals. The generalized symmetric group $S(m,n)$ is the wreath product of the cyclic group of order $m$ and the symmetric group of order $n!$. Our main focus is the underlying \emph{undirected} graphs, denoted by $\upmn$. In the cases when the cyclic group has one or two elements, these graphs are isomorphic to the pancake graphs and burnt pancake graphs, respectively. We prove that when the cyclic group has three elements, $\up3n$ has cycles of all possible lengths, thus resembling a similar property of pancake graphs and burnt pancake graphs. Moreover, $\up4n$ has all the even-length cycles. We utilize these results as base cases and show that if $m>2$ is even, $\upmn$ has all cycles of even length starting from its girth to a Hamiltonian cycle. Moreover, when $m>2$ is odd, $\upmn$ has cycles of all lengths starting from its girth to a Hamiltonian cycle. We furthermore show that the girth of $\upmn$ is $\min\{m,6\}$ if $m\geq3$, thus complementing the known results for $m=1,2.$

\end{abstract}

\maketitle

\section{Introduction}\label{s:intro}

The pancake problem consists of sorting a permutation utilizing prefix-reversals. In its classical formulation~\cite{Dweighter75}, one is tasked with sorting a stack of pancakes with different diameters utilizing only a chef's spatula by picking up some pancakes from the top of the stack, flipping them, and placing them back on top of the stack. The problem has attracted the interest of computer scientists due to its applications to interconnection networks~\cite{KF95} and the interest of computational biologists due to its application to genome rearrangements~\cite{HannenPev}. The burnt pancake problem, first introduced in~\cite{GatesPapa}, considers the setting where the elements of permutations also have a sign. In this setting, every time a prefix-reversal is applied to a permutation, the positions affected also have a change in sign (from ``$+$" to ``$-$" or the other way around.)

A natural extension of the \emph{symmetric group} of order $n!$, denoted by $S_n$, and the \emph{signed symmetric group} of order $2^n n!$, denoted by $B_n$, is the \emph{generalized symmetric group}, denoted by $\smn := C_m\wr S_n$, where $C_m$ is the cyclic group of order $m$ whereas the group $\smn$ is of order $m^n n!$. It follows that an extension of the pancake graphs and the burnt pancake graphs, along with the natural extension of prefix-reversals and their inverses, would be the Cayley graph of the generalized symmetric group generated by prefix-reversals. We refer to these graphs as the \emph{generalized pancake graphs} or \emph{$m$-sided pancake graphs}. Intuitively, we are dealing with a stack of $n$ pancakes, each with $m$ different ``sides," or colors. When $m=1$ and $m=2$, we recover the original setting of the pancake problem and burnt pancake problem, respectively. Each time one ``flips" a pancake, we cycle trough its sides, which resets after $m$ flips. Alternatively, each time one ``flops'' a pancake, we reverse the cycle of its sides. 

We prioritize the role of a ``flip'' in fidelity to the original \emph{pancake problem}. That is, we consider the directed graph with only edges between generalized permutations that are a ``flip'' away, denoted by $\pmn$. When including the set of ``flops'' in the generating set of the Cayley graph those edges may also be considered undirected, which in that case we denote the graph by $\upmn$. We focus our attention to the lengths of the cycles one can embed into $\pmn$ and $\upmn$. The existence of certain cycle lengths in $\pmn$ remains more open than in $\upmn$. For a graph $\Gamma$, we denote its vertex and edge set by $V(\Gamma)$ and $E(\Gamma)$, respectively.

\subsection{Main results} The main results of this article are the the following.
\begin{enumerate}
    \item We prove that the girth of the undirected generalized pancake graph $\upmn$ is $\min\{m,6\}$ if $m
    \geq3$, This result complements the cases when $m=1$, where the girth is $6$, and when $m=2$, where the girth is $8$, proved by Compeau~\cite[Theorem 5 and Theorem 10]{Compeau2011}. Consequently, the girth is known for all the undirected generalized pancake graphs.
    \item If $m\geq3$ is odd and $n\geq2$, then $\upmn$ has all cycles of length $\ell$ with $\min\{m,6\}\leq \ell\leq |V(\upmn)|$.
    \item If $m\geq3$ is even and $n\geq2$, then $\upmn$ has all cycles of length $2\ell$ with $\min\{m,6\}/2\leq \ell\leq |V(\upmn)|/2$. (Despite computer searches, no odd length cycles were found, see Section~\ref{sec:computer}.)
\end{enumerate} In Section~\ref{sec:prelims}, we list the main definitions and notation needed in the article. In Section~\ref{sec:pancakegraphs} we describe formally the generalized pancake graphs and their undirected counterparts. In Section~\ref{sec:girth}, we prove the girth results, providing a floor to cycle lengths. Moreover, in Sections~\ref{sec:up3n} and~\ref{sec:up4n} we prove that the $\up3n$ has all cycles of length $k$ with $3\leq \ell\leq |V(\up3n)|$ and that $\up4n$ has all cycles of length $2k$ with $2\leq k\leq |V(\up4n)|/2$. These are the base cases needed to prove the general results in Section~\ref{sec:generalcase}. Finally, in Section~\ref{sec:computer} we describe a parallel implementation, which we make available, of an algorithm in~\cite{HJ08} to find cycles in these generalized pancake graphs.

\section{Preliminaries}\label{sec:prelims}

Throughout this article, we shall use the \emph{integer interval} notation. Namely, if $i, j\in\mathbb{Z}$ with $i\leq j$ and $n\in\mathbb{Z}^+$, we define $[n] :=\{1,2,3, \ldots, n\}$ and $[i,j] := \{i, i+1, i+2, \ldots, j\}.$  Furthermore, let $\zeta$ be a primitive $m$-th root of unity. We will denote an element of $\smn$ by $\wt{\pi} = \vec{j} \wr \pi \in \smn$ as $\wt{\pi} = [\zeta^{j_1}\pi_1,\ldots,\zeta^{j_n}\pi_n]$ where $\vec{j} = \langle\zeta^{j_1},\zeta^{j_2},\ldots,\zeta^{j_n}\rangle$ with each $j_i\in [0,  m-1]$, for $i\in [n]$, and $\pi=\pi_1\pi_2\cdots\pi_n\in S_n$. 
Let $\wt{\pi}=[\zeta^{j_1}\pi_1,\zeta^{j_2}\pi_2,\ldots,\zeta^{j_n}\pi_n]$. Then, we define a \emph{generalized forward prefix reversal} or \emph{generalized pancake flip} of length $i \in [n]$ on $\wt{\pi}$ as 
   
\begin{align*}
r_i(\wt{\pi}):=[\zeta^{j_i+1}\pi_i,\zeta^{j_{i-1}+1}\pi_{i-1},\ldots,\zeta^{j_1+1}\pi_1,\zeta^{j_{i+1}}\pi_{i+1}, \zeta^{j_{i+2}}\pi_{i+2},\ldots,\zeta^{j_n}\pi_n)].
\end{align*}
The inverse of a flip of length $i$ is the \emph{generalized backward prefix reversal} or \emph{generalized pancake flop} of length $i \in [n]$ on $\wt{\pi}$, denoted by
\begin{align*}
\flop_i(\wt{\pi}):=[\zeta^{j_i-1}\pi_i,\zeta^{j_{i-1}-1}\pi_{i-1},\ldots,\zeta^{j_1-1}\pi_1,\zeta^{j_{i+1}}\pi_{i+1}, \zeta^{j_{i+2}}\pi_{i+2},\ldots,\zeta^{j_n}\pi_n)].
\end{align*}

For example, $r_2([\zeta^0 2,\zeta^2 1,\zeta^1 4,\zeta^2 3])=[\zeta^0 1, \zeta^1 2,\zeta^1 4, \zeta^2 3] \in S(3,4)$.

We will refer to the powers of $\zeta$ as \textit{signs} (also called ``colors" in~\cite{CSW21}) and the elements of the underlying permutation in $S_n$ as the \textit{symbols} of the generalized permutation.

For simplicity in the notation of elements of $\smn$, we will use one-line notation with superscripts to represent the signs of the elements. So in general, elements in $\smn$ will be denoted by $a_1^{i_1}a_2^{i_2}\cdots  a_n^{i_n}$, where $a_1,a_2,\ldots, a_n\in S_n$ and $i_1,i_2,\ldots,i_n\in [0,m-1]$.

For example, $4^23^11^02^1\in S(3,4)$ has four symbols and three signs (indicated by the superscripts) corresponds to the element $[\zeta^24,\zeta^13,\zeta^01,\zeta^12]$.

\subsection{Girth, pancyclic, panevencyclic}

By a \emph{cycle}, we mean a closed path where only the first and last vertices are repeated. We will focus on the nontrivial case where the cycles have length larger than two. 

A graph $\Gamma=(V,E)$ is said to be \emph{pancyclic} if it contains cycles of length $k$ for every $3\leq k\leq |V|$. Moreover, we say that $\Gamma$ is \emph{panevencyclic} if $|V|$ is even and if $\Gamma$ contains cycles of length $2k$ for every $2\leq k\leq |V|/2$. Here, $|V|$ denotes the number of vertices of $\Gamma$. In the same spirit, we say that $\Gamma$ is \emph{$k$-pancyclic} if $\Gamma$ has all cycles of length $\ell$ with $k\leq\ell\leq|V|$ and that $\Gamma$ is \emph{$2k$-panevencyclic} if $|V|$ is even and $\Gamma$ has all cycles of length $2\ell$ with $k\leq\ell\leq |V|/2$. Notice that when $\Gamma$ is panevencyclic, it is possible that it contains cycles of odd length.

The \emph{girth} of $\Gamma$ is the length of the shortest cycle of $\Gamma$. It is known that the girth of the pancake graph $\p1n=\up1n$ and the burnt pancake graph $\p2n=\up2n$ is six and eight, respectively~\cite{Compeau2011}. In the next section, we define formally the generalized pancake graphs $\pmn$ and their undirected counterparts $\upmn$.

\section{Generalized Pancake Graph}\label{sec:pancakegraphs}

For our purposes, if $G$ is a group and $S$ is a set that generates $G$, then the \emph{Cayley graph} of $G$ with respect to $S$ is the graph with vertex set $G$ and edge set $\{(g,sg)\colon s\in S, g\in G\}$.

With the collection of generalized prefix reversals, we can define what we call the \emph{directed generalized pancake graph}, denoted by $\pmn$, as the graph whose vertex set is $\smn$ and whose set of (directed) edges is given by
\[
\{\left(\wt{\pi},r_i(\wt{\pi})\right) \mid \wt{\pi}\in \smn, i \in [n]\},
\]
where we shall refer to the $r_i$ as the \emph{label} of the edge.

If the generating set contains both generators and their inverses, then the generating set is said to be \emph{symmetric}. The Cayley graph with vertex set $\smn$ and whose set of edges is given by
\begin{align*}
    \{\left(\wt{\pi},r_i(\wt{\pi})\right) \mid \wt{\pi}\in \smn, i \in [n]\} \cup
    \{\left(\wt{\pi},\flop_i(\wt{\pi})\right) \mid \wt{\pi}\in \smn, i \in [n]\},
\end{align*}
shall be referred to as the \emph{$m$-sided pancake graph} or \emph{undirected generalized pancake graph}, denoted by $\upmn$. When the generating set is symmetric the resulting Cayley graph is commonly considered an undirected graph. We note that the generating set $\{r_i, \flop_i \mid i \in [n]\}$ is symmetric and that $\upmn$ is the underlying, undirected graph of $\pmn$ and is the Cayley graph of $\smn$ first studied by Justan, Muga, and Sudborough in~\cite{JMS02}. 
In our study, the following two properties are useful.
\begin{defn}[Vertex- and edge-transitive]
Given a graph $\Gamma=(V,E)$, we say that $\Gamma$ is 
\begin{description}
    \item[vertex-transitive] If for any pair $v_1,v_2\in V$, there exists some $\varphi\in \Aut(\Gamma)$ such that $\varphi(v_1)=v_2$. Moreover, if $V'\subseteq V$ and for every $v'_1,v'_2\in V'$, there exists $\varphi'\in\Aut(\Gamma)$ such that $\varphi'(v'_1)=v'_2$, we say that $V'$ is \emph{vertex-transitive}. 
    \item[edge-transitive] If for any pair $e_1,e_2\in E$, there exists some $\psi\in \Aut(\Gamma)$ such that $\psi(e_1)=e_2$. Moreover, if $E'\subseteq E$ and for every $e'_1,e'_2\in E'$, there exists $\psi'\in\Aut(\Gamma)$ such that $\psi'(e'_1)=e'_2$, we say that $E'$ is \emph{edge-transitive}. 
\end{description} 
\end{defn}

Since $\pmn$ and $\upmn$ are Cayley graphs, it follows that they are vertex-transitive. However, $\pmn$ and $\upmn$ are not edge-transitive but the set of edges with the same label do form an edge-transitive set (see Figure~\ref{fig:\p32} and Figure~\ref{fig:\p42}.) We prove this property in the Lemma that follows. 

\begin{lem}\label{l:edgetrans}
 Let $E^k_n$ be the set of edges of $\pmn$ labeled by 
 $r_k$, with $k\in [n]$. Then, for every pair $e_1,e_2\in E^k_n$, there exists $\varphi\in \Aut(\pmn)$ such that $\varphi(e_1)=e_2$.
\end{lem}
\begin{proof}
Let $e_1,e_2\in E^k_n$. Then the edges $e_1,e_2$ have the form 
$(\wt{\pi}_1,r_k(\wt{\pi}_1))$ and $(\wt{\pi}_2,r_k(\wt{\pi}_2))$, with $\wt{\pi}_1,\wt{\pi}_2\in \smn$. Let $\varphi:V(\pmn)\to V(\pmn)$ be given by $\varphi(x)=x\wt{\pi}_1^{-1}\wt{\pi}_2$. Since $\varphi$ is defined as right multiplication by a specific group element of $S(m,n)$, the underlying group of the Cayley graph $\pmn$, then $\varphi$ is a graph automorphism of $\pmn$. Indeed, $(\wt{\pi},\wt{\pi}')\in E(\pmn)$ if and only if $\wt{\pi}'=r_i(\wt{\pi})$ for some $i\in [n]$, and if and only if $(\varphi(\wt{\pi}),\varphi(\wt{\pi}'))\in E(\pmn)$. Moreover,  $\varphi(\wt{\pi}_1)=\wt{\pi}_2$ and 
$\varphi(r_k(\wt{\pi}_1))=r_k(\wt{\pi}_2)$, and so $\varphi(e_1)=e_2$.
\end{proof} 

One of the questions that has been of interest relating to pancake graphs regards the size of cycles that can be embedded in these graphs. For instance, both the pancake and the burnt pancake graphs are \emph{weakly pancyclic}, that is, they are $6$-pancyclic and $8$-pancyclic, respectively~\cite{BBP19,KF95}. Moreover, small cycles have been classified for pancake graphs in~\cite{KM11,KM10,KM16} and for  burnt pancake graphs in~\cite{BBP19, BBP19Perm}. Recent results demonstrate a Hamiltonian cycle within the generalized pancake graphs~\cite{CSW21}.

The property of weakly-pancyclic has been of interest due to their relation to interconnection networks, see~\cite{KF95}. Unfortunately, the lengths of cycles that can be embedded in $\pmn$ seems to be a more idiosyncratic with a pattern that is harder to discern, as shown in Table~\ref{tab:directed-cycles}.

\begin{table}[ht!]
\begin{tabular}{|r|r|l|}
\hline
\multicolumn{1}{|l|}{$m$} & \multicolumn{1}{l|}{$n$} & length of cycles \\ \hline
3                       & 2                 & $[3,18] \smallsetminus \{5,7,11,16,17\}$\\ \hline 
4                       & 2                 & $\{2k \mid k \in [2,16]\}\smallsetminus\{30\}$\\ \hline 
5                       & 2                      &$[5,50] \smallsetminus \{7,9,13,47,48,49\}$\\ \hline 
6                       & 2                      & $\{2k \mid k\in[2,36]\}$                   \\ \hline
\end{tabular}
\caption{Lengths of cycles in $\pmn$ for different values of $m,n$. Notice that $\p{4}{2}$ has all cycles of even length except for 30}
\label{tab:directed-cycles}
\end{table}

In the undirected case, $\upmn$, we do find that there is a clear pattern for the lengths of the cycles in some of these graphs, as shown in Table~\ref{tab:undirected-cycles}. For example, we will show $\up3n$ is pancyclic. 

\begin{table}[ht!]
\begin{tabular}{|r|r|l|}
\hline
\multicolumn{1}{|l|}{$m$} & \multicolumn{1}{l|}{$n$} & length of cycles \\ \hline
2                       & 3                      & $[8,48]$\\ \hline
3                       & 2                      & $[3,18]$\\ \hline
3                       & 3                      & $[3,162]$\\ \hline
4                       & 2                      & $\{2k \mid k \in [2,16]\}$\\ \hline 
5                       & 2                      & $[5,50]$\\ \hline
\end{tabular}
\caption{Lengths of cycles in $\upmn$ for different values of $m,n$.}
\label{tab:undirected-cycles}
\end{table}

\pgfplotsset{compat=1.15}
\usetikzlibrary{arrows}
\definecolor{qqttqq}{rgb}{0.,0.2,0.}
\definecolor{qqwuqq}{rgb}{0.,0.39215686274509803,0.}
\definecolor{ccqqqq}{rgb}{0.8,0,0}
\definecolor{ududff}{rgb}{0.30196078431372547,0.30196078431372547,1}
\begin{figure}
    \centering
\begin{tikzpicture}[line cap=round,line join=round,>=triangle 45,x=1cm,y=1cm]
\draw [->,line width=0.5pt,color=ccqqqq] (-12.,6.) -- (-10.,6.);
\draw [->,line width=0.5pt,color=ccqqqq,dashed] (-10.,6.) -- (-11.,8.);
\draw [->,line width=0.5pt,color=ccqqqq] (-11.,8.) -- (-12.,6.);
\draw [->,line width=0.5pt,color=qqwuqq] (-14.,7.) -- (-16.,6.);
\draw [->,line width=0.5pt,color=ccqqqq] (-14.,7.) -- (-13.,9.);
\draw [->,line width=0.5pt,color=ccqqqq,dashed] (-13.,9.) -- (-15.,9.);
\draw [->,line width=0.5pt,color=ccqqqq] (-15.,9.) -- (-14.,7.);
\draw [->,line width=0.5pt,color=ccqqqq] (-16.,6.) -- (-17.,8.);
\draw [->,line width=0.5pt,color=ccqqqq,dashed] (-17.,8.) -- (-18.,6.);
\draw [->,line width=0.5pt,color=ccqqqq] (-18.,6.) -- (-16.,6.);
\draw [->,line width=0.5pt,color=ccqqqq] (-16.,4.) -- (-18.,4.);
\draw [->,line width=0.5pt,color=ccqqqq,dashed] (-18.,4.) -- (-17.,2.);
\draw [->,line width=0.5pt,color=ccqqqq] (-17.,2.) -- (-16.,4.);
\draw [->,line width=0.5pt,color=ccqqqq] (-14.,3.) -- (-15.,1.);
\draw [->,line width=0.5pt,color=ccqqqq,dashed] (-15.,1.) -- (-13.,1.);
\draw [->,line width=0.5pt,color=ccqqqq] (-13.,1.) -- (-14.,3.);
\draw [->,line width=0.5pt,color=ccqqqq] (-12.,4.) -- (-11.,2.);
\draw [->,line width=0.5pt,color=ccqqqq,dashed] (-11.,2.) -- (-10.,4.);
\draw [->,line width=0.5pt,color=ccqqqq] (-10.,4.) -- (-12.,4.);
\draw [->,line width=0.5pt,color=qqttqq,dashed] (-10.,6.) -- (-10.,4.);
\draw [->,line width=0.5pt,color=qqttqq] (-10.,4.) -- (-17.,8.);
\draw [->,line width=0.5pt,color=qqttqq,dashed] (-17.,8.) -- (-15.,9.);
\draw [->,line width=0.5pt,color=qqttqq] (-15.,9.) -- (-15.,1.);
\draw [->,line width=0.5pt,color=qqttqq,dashed] (-15.,1.) -- (-17.,2.);
\draw [->,line width=0.5pt,color=qqttqq] (-17.,2.) -- (-10.,6.);
\draw [->,line width=0.5pt,color=qqttqq,dashed] (-11.,2.) -- (-13.,1.);
\draw [->,line width=0.5pt,color=qqttqq] (-13.,1.) -- (-13.,9.);
\draw [->,line width=0.5pt,color=qqttqq,dashed] (-13.,9.) -- (-11.,8.);
\draw [->,line width=0.5pt,color=qqttqq] (-11.,8.) -- (-18.,4.);
\draw [->,line width=0.5pt,color=qqttqq,dashed] (-18.,4.) -- (-18.,6.);
\draw [->,line width=0.5pt,color=qqttqq] (-18.,6.) -- (-11.,2.);
\draw [->,line width=0.5pt,color=qqwuqq] (-12.,6.) -- (-14.,7.);
\draw [->,line width=0.5pt,color=qqwuqq] (-16.,6.) -- (-16.,4.);
\draw [->,line width=0.5pt,color=qqwuqq] (-16.,4.) -- (-14.,3.);
\draw [->,line width=0.5pt,color=qqwuqq] (-14.,3.) -- (-12.,4.);
\draw [->,line width=0.5pt,color=qqwuqq] (-12.,4.) -- (-12.,6.);
\begin{scriptsize}
\draw [fill=ududff] (-12.,6.) circle (2.5pt);
\draw[color=ududff] (-12.470927732826164,5.960200078854992) node {$2^21^0$};
\draw [fill=ududff] (-10.,6.) circle (2.5pt);
\draw[color=ududff] (-9.615837700803436,6.219133048086954) node {$2^01^0$};
\draw [fill=ududff] (-11.,8.) circle (2.5pt);
\draw[color=ududff] (-10.779022528664548,8.31678232750867) node {$2^11^0$};
\draw [fill=ududff] (-13.,1.) circle (2.5pt);
\draw[color=ududff] (-12.727734772743553,0.6600159486987923) node {$2^21^1$};
\draw [fill=ududff] (-17.,8.) circle (2.5pt);
\draw[color=ududff] (-17.33515519479081,8.237632) node {$2^21^2$};
\draw [fill=ududff] (-14.,3.) circle (2.5pt);
\draw[color=ududff] (-14.087301454659137,3.427661162009757) node {$2^01^1$};
\draw [fill=ududff] (-15.,1.) circle (2.5pt);
\draw[color=ududff] (-15.326017764848892,0.6146970593016062) node {$2^11^1$};
\draw [fill=ududff] (-16.,6.) circle (2.5pt);
\draw[color=ududff] (-15.524100362355025,5.929987485923535) node {$2^11^2$};
\draw [fill=ududff] (-18.,6.) circle (2.5pt);
\draw[color=ududff] (-18.468127429720465,6.152963786963938) node {$2^01^2$};
\draw [fill=ududff] (-17.,2.) circle (2.5pt);
\draw[color=ududff] (-17.35026149125654,1.6872441083683439) node {$1^22^2$};
\draw [fill=ududff] (-16.,4.) circle (2.5pt);
\draw[color=ududff] (-15.619206658820753,4.313613764090564) node {$1^02^2$};
\draw [fill=ududff] (-18.,4.) circle (2.5pt);
\draw[color=ududff] (-18.447277057994635,3.966168945378804) node {$1^12^2$};
\draw [fill=ududff] (-13.,9.) circle (2.5pt);
\draw[color=ududff] (-12.939222923263754,9.37422308010968) node {$1^22^0$};
\draw [fill=ududff] (-10.,4.) circle (2.5pt);
\draw[color=ududff] (-9.630943997269165,3.66404301606423) node {$1^12^1$};
\draw [fill=ududff] (-15.,9.) circle (2.5pt);
\draw[color=ududff] (-14.99367924260286,9.37422308010968) node {$1^02^0$};
\draw [fill=ududff] (-11.,2.) circle (2.5pt);
\draw[color=ududff] (-10.733703639267361,1.807669294231259) node {$1^02^1$};
\draw [fill=ududff] (-14.,7.) circle (2.5pt);
\draw[color=ududff] (-14.102407751124865,6.6910463844700715) node {$1^12^0$};
\draw [fill=ududff] (-12.,4.) circle (2.5pt);
\draw[color=ududff] (-12.470927732826164,4.35893265348775) node {$1^22^1$};
\end{scriptsize}
\end{tikzpicture}
    \caption{Generalized pancake graph with three signs and 2 symbols, $\p32$. The ``base cycle," starting from $1^02^0$, is depicted with the dashed edges and the red edges depict the six copies of $\p31$, all of them a directed 3-cycle, embedded into $\p32$. A computer search shows that there is not a 5-cycle in $\p32$ (or cycles of length 7, 11, and 17.) Nonetheless, $\up32$ is pancyclic.}
    \label{fig:\p32}
\end{figure}
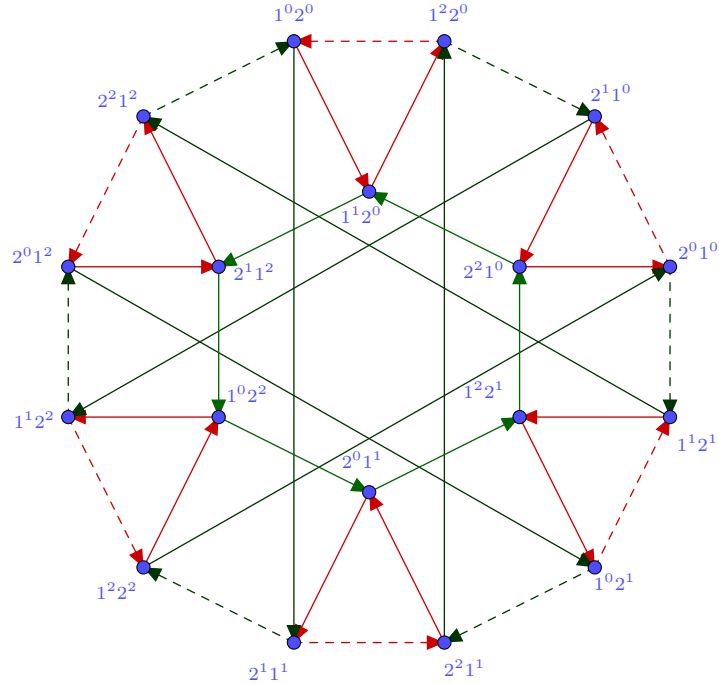

\definecolor{ccqqqq}{rgb}{0.8,0,0}
\definecolor{wrwrwr}{rgb}{0.3803921568627451,0.3803921568627451,0.3803921568627451}
\definecolor{rvwvcq}{rgb}{0.08235294117647059,0.396078431372549,0.7529411764705882}

\begin{figure}
\centering
\begin{tikzpicture}[line cap=round,line join=round,>=triangle 45,x=1cm,y=1cm,scale=0.91]
\draw [->,line width=0.5pt,dash pattern=on 3pt off 3pt,color=wrwrwr] (2.,-6.) -- (2.,-4.);
\draw [->,line width=0.5pt,dash pattern=on 3pt off 3pt,color=wrwrwr] (4.,-4.) -- (4.,-2.);
\draw [->,line width=0.5pt,dash pattern=on 3pt off 3pt,color=wrwrwr] (4.,0.) -- (2.,0.);
\draw [->,line width=0.5pt,dash pattern=on 3pt off 3pt,color=wrwrwr] (2.,2.) -- (0.,2.);
\draw [->,line width=0.5pt,dash pattern=on 3pt off 3pt,color=wrwrwr] (-2.,2.) -- (-2.,0.);
\draw [->,line width=0.5pt,dash pattern=on 3pt off 3pt,color=wrwrwr] (-4.,0.) -- (-4.,-2.);
\draw [->,line width=0.5pt,dash pattern=on 3pt off 3pt,color=wrwrwr] (-4.,-4.) -- (-2.,-4.);
\draw [->,line width=0.5pt,dash pattern=on 3pt off 3pt,color=wrwrwr] (-2.,-6.) -- (0.,-6.);
\draw [->,line width=0.5pt,color=ccqqqq] (-2.,2.) -- (-2.,4.);
\draw [->,line width=0.5pt,color=wrwrwr] (-2.,0.) -- (2.,-6.);
\draw [->,line width=0.5pt,color=wrwrwr] (-4.,-2.) -- (4.,-4.);
\draw [->,line width=0.5pt,color=wrwrwr] (-6.,2.) -- (0.,4.);
\draw [->,line width=0.5pt,color=wrwrwr] (-2.,4.) -- (4.,2.);
\draw [->,line width=0.5pt,color=wrwrwr] (0.,4.) -- (6.,-6.);
\draw [->,line width=0.5pt,color=wrwrwr] (-4.,2.) -- (6.,0.);
\draw [->,line width=0.5pt,color=wrwrwr] (-6.,-2.) -- (4.,4.);
\draw [->,line width=0.5pt,color=wrwrwr] (-6.,-4.) -- (-4.,2.);
\draw [->,line width=0.5pt,color=wrwrwr] (-4.,-6.) -- (-2.,4.);
\draw [->,line width=0.5pt,color=wrwrwr] (-4.,-8.) -- (-6.,-2.);
\draw [->,line width=0.5pt,color=wrwrwr] (0.04344806595431746,-8.04547839077213) -- (-6.,2.);
\draw [->,line width=0.5pt,color=wrwrwr] (2.,-8.) -- (-4.,-6.);
\draw [->,line width=0.5pt,color=wrwrwr] (6.,-6.) -- (0.04344806595431727,-8.04547839077213);
\draw [->,line width=0.5pt,color=wrwrwr] (4.,-6.) -- (-6.,-4.);
\draw [->,line width=0.5pt,color=wrwrwr] (6.,-2.) -- (-4.,-8.);
\draw [->,line width=0.5pt,color=wrwrwr] (6.,0.) -- (4.,-6.);
\draw [->,line width=0.5pt,color=wrwrwr] (4.,2.) -- (2.,-8.);
\draw [->,line width=0.5pt,color=wrwrwr] (4.,4.) -- (6.,-2.);
\draw [->,line width=0.5pt,color=wrwrwr] (2.,-4.) -- (-2.,2.);
\draw [->,line width=0.5pt,color=wrwrwr] (4.,-2.) -- (-4.,0.);
\draw [->,line width=0.5pt,color=wrwrwr] (2.,0.) -- (-4.,-4.);
\draw [->,line width=0.5pt,color=wrwrwr] (0.,-6.) -- (2.,2.);
\draw [->,line width=0.5pt,color=wrwrwr] (-2.,-4.) -- (4.,0.);
\draw [->,line width=0.5pt,color=wrwrwr] (0.,2.) -- (-2.,-6.);
\draw [->,line width=0.5pt,dash pattern=on 3pt off 3pt,color=ccqqqq] (-4.,-4.) -- (-4.,-2.);
\draw [->,line width=0.5pt,dash pattern=on 3pt off 3pt,color=ccqqqq] (-4.,0.) -- (-2.,0.);
\draw [->,line width=0.5pt,dash pattern=on 3pt off 3pt,color=ccqqqq] (-2.,2.) -- (0.,2.);
\draw [->,line width=0.5pt,dash pattern=on 3pt off 3pt,color=ccqqqq] (2.,2.) -- (2.,0.);
\draw [->,line width=0.5pt,dash pattern=on 3pt off 3pt,color=ccqqqq] (4.,0.) -- (4.,-2.);
\draw [->,line width=0.5pt,dash pattern=on 3pt off 3pt,color=ccqqqq] (4.,-4.) -- (2.,-4.);
\draw [->,line width=0.5pt,dash pattern=on 3pt off 3pt,color=ccqqqq] (2.,-6.) -- (0.,-6.);
\draw [->,line width=0.5pt,dash pattern=on 3pt off 3pt,color=ccqqqq] (-2.,-6.) -- (-2.,-4.);
\draw [->,line width=0.5pt,color=ccqqqq] (0.,2.) -- (0.,4.);
\draw [->,line width=0.5pt,color=ccqqqq] (0.,4.) -- (-2.,4.);
\draw [->,line width=0.5pt,color=ccqqqq] (-2.,0.) -- (-4.,2.);
\draw [->,line width=0.5pt,color=ccqqqq] (-4.,2.) -- (-6.,2.);
\draw [->,line width=0.5pt,color=ccqqqq] (-6.,2.) -- (-4.,0.);
\draw [->,line width=0.5pt,color=ccqqqq] (-4.,-2.) -- (-6.,-2.);
\draw [->,line width=0.5pt,color=ccqqqq] (-6.,-2.) -- (-6.,-4.);
\draw [->,line width=0.5pt,color=ccqqqq] (-6.,-4.) -- (-4.,-4.);
\draw [->,line width=0.5pt,color=ccqqqq] (-2.,-4.) -- (-4.,-6.);
\draw [->,line width=0.5pt,color=ccqqqq] (-4.,-6.) -- (-4.,-8.);
\draw [->,line width=0.5pt,color=ccqqqq] (-4.,-8.) -- (-2.,-6.);
\draw [->,line width=0.5pt,color=ccqqqq] (0.,-6.) -- (0.04344806595431746,-8.04547839077213);
\draw [->,line width=0.5pt,color=ccqqqq] (0.04344806595431746,-8.04547839077213) -- (2.,-8.);
\draw [->,line width=0.5pt,color=ccqqqq] (2.,-8.) -- (2.,-6.);
\draw [->,line width=0.5pt,color=ccqqqq] (2.,-4.) -- (4.,-6.);
\draw [->,line width=0.5pt,color=ccqqqq] (4.,-6.) -- (6.,-6.);
\draw [->,line width=0.5pt,color=ccqqqq] (6.,-6.) -- (4.,-4.);
\draw [->,line width=0.5pt,color=ccqqqq] (4.,-2.) -- (6.,-2.);
\draw [->,line width=0.5pt,color=ccqqqq] (6.,-2.) -- (6.,0.);
\draw [->,line width=0.5pt,color=ccqqqq] (6.,0.) -- (4.,0.);
\draw [->,line width=0.5pt,color=ccqqqq] (2.,0.) -- (4.,2.);
\draw [->,line width=0.5pt,color=ccqqqq] (4.,2.) -- (4.,4.);
\draw [->,line width=0.5pt,color=ccqqqq] (4.,4.) -- (2.,2.);
\begin{scriptsize}
\draw [fill=rvwvcq] (0.,-6.) circle (2.5pt);
\draw[color=rvwvcq] (0.3302818931299044,-6.23948762707398) node {$1^32^2$};
\draw [fill=rvwvcq] (2.,-4.) circle (2.5pt);
\draw[color=rvwvcq] (1.5288885507041743,-4.0385822588635037) node {$2^31^3$};
\draw [fill=rvwvcq] (4.,-2.) circle (2.5pt);
\draw[color=rvwvcq] (4.424708523426984,-1.6801707909753365) node {$1^02^3$};
\draw [fill=rvwvcq] (2.,0.) circle (2.5pt);
\draw[color=rvwvcq] (1.981309254894196,-0.3903415819506756) node {$2^01^0$};
\draw [fill=rvwvcq] (0.,2.) circle (2.5pt);
\draw[color=rvwvcq] (-0.324047142324131,2.3779966898052034) node {$1^12^0$};
\draw [fill=rvwvcq] (-2.,0.) circle (2.5pt);
\draw[color=rvwvcq] (-1.4757088705682435,0.09210352095010826) node {$2^11^1$};
\draw [fill=rvwvcq] (-4.,-2.) circle (2.5pt);
\draw[color=rvwvcq] (-4.577763594096827,-2.2238140466157494) node {$1^22^1$};
\draw [fill=rvwvcq] (-2.,-4.) circle (2.5pt);
\draw[color=rvwvcq] (-2.1193521262086575,-3.6074085048346373) node {$2^21^2$};
\draw [fill=rvwvcq] (4.,-4.) circle (2.5pt);
\draw[color=rvwvcq] (3.6660530684311894,-4.278533461441979) node {$2^21^3$};
\draw [fill=rvwvcq] (4.,0.) circle (2.5pt);
\draw[color=rvwvcq] (4.388449373910448,-0.24784768162836632) node {$1^32^3$};
\draw [fill=rvwvcq] (2.,2.) circle (2.5pt);
\draw[color=rvwvcq] (2.633729959084218,2.046822935776337) node {$2^31^0$};
\draw [fill=rvwvcq] (-2.,2.) circle (2.5pt);
\draw[color=rvwvcq] (-1.5881783721797903,2.2505149888382753) node {$1^02^0$};
\draw [fill=rvwvcq] (-4.,0.) circle (2.5pt);
\draw[color=rvwvcq] (-3.812873388295259,0.4170425241732014) node {$2^01^1$};
\draw [fill=rvwvcq] (-4.,-4.) circle (2.5pt);
\draw[color=rvwvcq] (-4.441504444580291,-3.719878006446184) node {$1^12^1$};
\draw [fill=rvwvcq] (2.,-6.) circle (2.5pt);
\draw[color=rvwvcq] (1.513876351348793,-6.196993726751671) node {$1^22^2$};
\draw [fill=rvwvcq] (-2.,-6.) circle (2.5pt);
\draw[color=rvwvcq] (-2.6442911294317506,-5.823326072400495) node {$2^11^2$};
\draw [fill=rvwvcq] (0.04344806595431746,-8.04547839077213) circle (2.5pt);
\draw[color=rvwvcq] (-0.28587966154358135,-8.300441792705984) node {$1^02^2$};
\draw [fill=rvwvcq] (2.,-8.) circle (2.5pt);
\draw[color=rvwvcq] (2.312483008923063,-8.300441792705984) node {$1^12^2$};
\draw [fill=rvwvcq] (4.,-6.) circle (2.5pt);
\draw[color=rvwvcq] (3.687300018592344,-6.296993726751671) node {$2^01^3$};
\draw [fill=rvwvcq] (6.,-6.) circle (2.5pt);
\draw[color=rvwvcq] (6.3006748884143695,-5.559585221917031) node {$2^11^3$};
\draw [fill=rvwvcq] (6.,-2.) circle (2.5pt);
\draw[color=rvwvcq] (6.4006748884143695,-1.6439116414588006) node {$1^12^3$};
\draw [fill=rvwvcq] (6.,0.) circle (2.5pt);
\draw[color=rvwvcq] (6.470650489703607,0.13459742127241758) node {$1^22^3$};
\draw [fill=rvwvcq] (-6.,-2.) circle (2.5pt);
\draw[color=rvwvcq] (-6.487446410856914,-1.8413689437149658) node {$1^32^1$};
\draw [fill=rvwvcq] (-6.,-4.) circle (2.5pt);
\draw[color=rvwvcq] (-6.444952510534604,-3.987310909991586) node {$1^02^1$};
\draw [fill=rvwvcq] (4.,2.) circle (2.5pt);
\draw[color=rvwvcq] (4.382214623104675,1.855600384325945) node {$2^11^0$};
\draw [fill=rvwvcq] (4.,4.) circle (2.5pt);
\draw[color=rvwvcq] (4.303461573265829,4.267728303986813) node {$2^21^0$};
\draw [fill=rvwvcq] (0.,4.) circle (2.5pt);
\draw[color=rvwvcq] (0.3090349429687497,4.267728303986813) node {$1^22^0$};
\draw [fill=rvwvcq] (-2.,4.) circle (2.5pt);
\draw[color=rvwvcq] (-1.9218948239524918,4.263963054792587) node {$1^32^0$};
\draw [fill=rvwvcq] (-4.,2.) circle (2.5pt);
\draw[color=rvwvcq] (-3.6853916873283303,2.3505149888382753) node {$2^21^1$};
\draw [fill=rvwvcq] (-6.,2.) circle (2.5pt);
\draw[color=rvwvcq] (-6.48121166005114,2.1318107364209555) node {$2^31^1$};
\draw [fill=rvwvcq] (-4.,-6.) circle (2.5pt);
\draw[color=rvwvcq] (-4.471528843291053,-5.410856570788948) node {$2^31^2$};
\draw [fill=rvwvcq] (-4.,-8.) circle (2.5pt);
\draw[color=rvwvcq] (-4.131577640712579,-8.279194842544829) node {$2^01^2$};
\end{scriptsize}
\end{tikzpicture}
    \caption{Generalized pancake graph with four signs and 2 symbols, $\p{4}{2}$. The ``base cycle," starting from $1^02^0$, is shown in dashed lines and in red edges show the embedded copies of $\p{4}{1}$. A computer search verifies that $\p{4}{2}$ has all cycles of length $4,6,\ldots,32$, except for a cycle of length 30. The corresponding undirected graph does have a cycle of length 30.}
    \label{fig:\p42}
\end{figure}
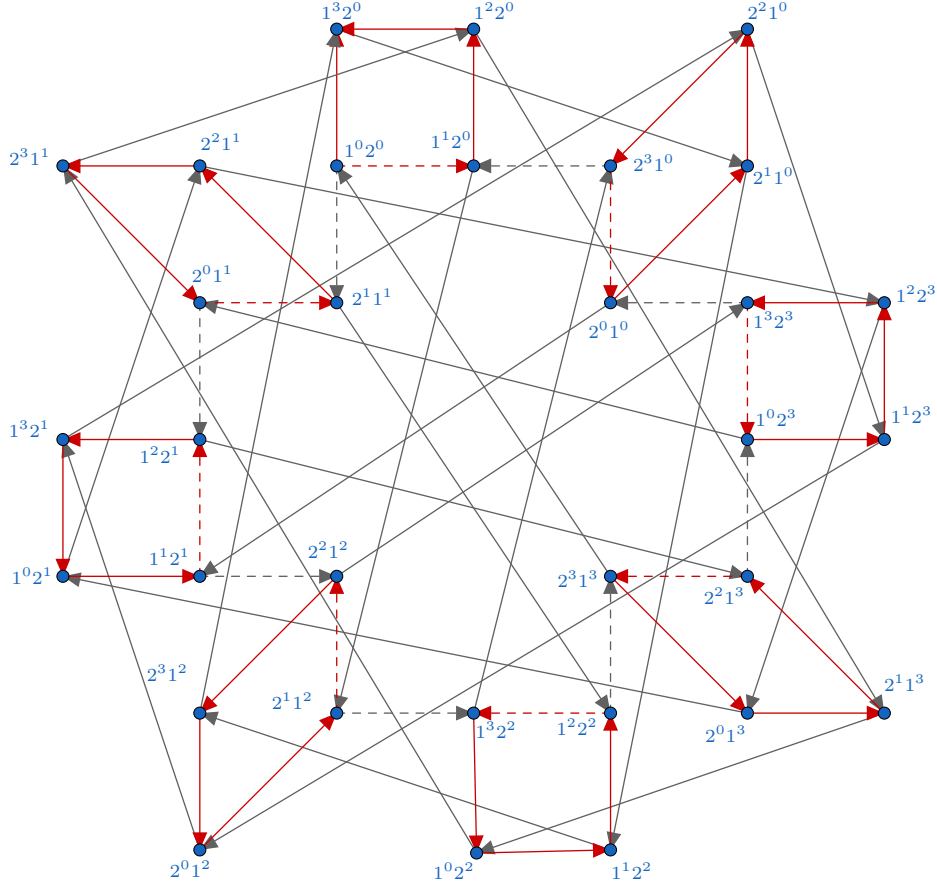

\subsection{Hierarchical structure of \texorpdfstring{$\pmn$}{P(m,n)} and \texorpdfstring{$\upmn$}{Pm(n)}}\label{sec:hier}

One of the more remarkable features of the pancakes graphs (see~\cite{BBP19, KF95}) is their hierarchical structure. If one considers the induced subgraph whose vertices share a fixed last symbol (which is equivalent to thinking of all the stacks with the same bottom pancake), then the resulting induced graph is isomorphic to the pancake graph with $n-1$ symbols. In a similar fashion, by looking at the induced subgraph obtained from $\pmn$ and $\upmn$ by leaving the last signed-symbol fixed, we obtain a graph isomorphic to $\p{m}{n-1}$ and $\up{m}{n-1}$, respectively. To illustrate the concept, let us observe that in Figure~\ref{fig:\p32}, there are six copies of $\p{3}{1}$ and $\up{3}{1}$ (they appear as a directed and undirected triangle) inside of $\p{3}{2}$ and $\up{3}{2}$, respectively. In general, since there are $mn$ potential last signed-symbols, then there are $mn$ non-overlapping copies of $\p{m}{n-1}$ and $\up{m}{n-1}$ embedded into $\pmn$ and $\upmn$, with $n\geq2$. The only edges connecting between these copies would be the edges labeled by $r_n$ or $\flop_n$.

There is a particular cycle connecting all these copies of $\up{m}{n-1}$ in $\upmn$ that is of interest. This cycle will be relevant when we characterize the length of the cycles in $\up{3}{n}$ and $\up{4}{n}$. We shall refer to this cycle, as the \emph{base cycle}.

\subsection{Base cycle} For $n>2$ and referring to the edge labels, we define the cycle
 
\[\C:=(r_{n-1}\flop_n)^{mn}.\]

\begin{lem}\label{lem:basecycle}
$\C$ is a cycle in $\upmn$ of length $2mn$.
\end{lem}
\begin{proof}
That the length of the path  $(r_{n-1}\flop_n)^{mn}$ is $2mn$ is straightforward. 

To show that $(r_{n-1}\flop_n)^{mn}$ is indeed a cycle, let us consider the effect of applying $(r_{n-1}\flop_n)^{n}$ to the identity $e=1^02^0\cdots n^0 \in \smn$. Notice that the effect on the signed-symbols of $e$ after applying $(r_{n-1}\flop_n)^{i}$, with $i \in [n]$, is
\[
    (r_{n-1}\flop_{n})^{i}(e) = (i+1)^0(i+2)^0 \cdots n^0 1^{m-1}2^{m-1} \cdots i^{m-1}
\]
This is because $\flop_n$ reverses all the symbols while decrementing their signs and then $r_{n-1}$ reverses the first $n-1$ symbols and increases their signs back to zero. Thus, after each flop-flip the original first symbol is at the end of the generalized permutation with the sign reduced by one while the other $n-1$ symbols move to the left in their same relative order and original sign. Therefore, after applying $r_{n-1}\flop_n$ a total of $n$ times beginning with the generalized permutation $e$ the symbols will end up in their original position while the signs of each symbol are now $m-1$. Thus to cycle back to the original sign of zero requires $m$ iterations of $(r_{n-1}\flop_n)^n$ to $e$.
 
Therefore, the order of $r_{n-1}\flop_n$ is $mn$. 
\end{proof} We depict the base cycles for $\p32$ and $\p42$ in Figure~\ref{fig:\p32} and~\ref{fig:\p42}.

\section{The girth of \texorpdfstring{$\up{m}{n}$}{Pm(n)}}\label{girth}\label{sec:girth}

In this section, we study the girth of the undirected generalized pancake graphs. We shall consider the cases when $n=2$ and when $n>2$ separately. These two cases are distinguished by the girths of their quotient graphs. For a given graph $\Gamma=(V,E)$ and an equivalence relation $\sim$ on $V$, the \emph{quotient graph} $\Gamma/{\sim}$ is the graph whose vertex set is the quotient set $V/{\sim}$ and whose edge set is $\{([u],[v]) \mid (u,v) \in E\}$. That is vertices in the quotient graph are connected if any of the vertices in two equivalence classes have an edge between them. 

A natural quotient graph to consider amongst generalized pancake graphs is the one with respect to the equivalence relation $\sim_n$ that identifies generalized permutations with the same last signed-symbol, that is the quotient would be modulo the $mn$-copies of $\up{m}{n-1}$ that are found in $\up{m}{n}$. We shall denote this quotient graph by $\up{m}{n}/\up{m}{n-1}$. As noted in Section~\ref{sec:hier}, there are $mn$ vertices in this quotient graph and the edges connecting would correspond to edges labeled by $r_n$ and $\flop_n$. Furthermore, for $n \geq 2$, whenever the last symbols matches in two generalized permutations there is not a single flip/flop of length $n$ that can connect them, thus the quotient graph will be a complete multipartite graph where there are $n$ parts, those with matching last symbol, with $m$ vertices within each part, one for each sign. It is well known that the girth of a complete bipartite graph with equal size parts is four. It is also straightforward to see that the girth of any complete multipartite graph, with more than two parts, is three.

The following is a useful observation connecting cycles in $\up{m}{n}$ and the quotient graph $\up{m}{n} /\up{m}{n-1}$.

\begin{lem}\label{lem:quotient}
If $C=(V_C,E_C)$ is a cycle in $\up{m}{n}$ consisting of edges labeled by $R_n$, where $R_n=\{r_n,\flop_n\}$, then $C'=(V_C/\sim_n,\{([u],[v])\mid u,v\in E_C\})$ is a cycle in $\up{m}{n} /\up{m}{n-1}$.
\end{lem}

\begin{proof}
Let us denote $C$ by a sequence of vertices connected by edges where the first and last vertex are the same.  Without loss of generality, $C$ has the form $C=v_0R_nS_1R_n\cdots S_kR_nv_0$ where each $S_i$ is a sequence of vertices and edges from\[\{r_1,\flop_1,\ldots,r_{n-1},\flop_{n-1}\}\cup S(m,n).\] Since $C$ is a cycle, each $S_i$ lies in a different copy of $\up{m}{n-1}$ embedded in $\up{m}{n}$. Under the equivalence relation $\sim_n$, each $S_i$ becomes a vertex $v_{i}$ of $\up{m}{n}/\up{m}{n-1}$. Since each $C$ is a cycle, $S_i$ and $S_j$ belong to different copies of $\up{m}{n-1}$ embedded in $\up{m}{n}$ if $i\neq j$. Therefore, $v_{i}\neq v_{j}$ if $i\neq j$, and therefore $C'=v_0R_nv_1R_n\cdots v_kR_nv_0$ is a cycle in $\up{m}{n} /\up{m}{n-1}$.
\end{proof}

\begin{thm}\label{t:n=2girths}
For $m \geq 3$, the girth of $\up{m}{2}$ is $\min\{m,6\}$.
\end{thm}
\begin{proof}
We consider the case $m\leq 5$ first and consider a cycle $C$ in $\up{m}{2}$. If $C$ only uses edges labeled by $r_1$ or $\flop_1$, then $C$ is embedded in one of the copies of $\up{m}{1}$ in $\up{m}{2}$, and so $|C|=m$. Moreover, if $C$ only traverses edges labeled by $r_2$ or $\flop_2$, then $|C|$ is $m$ or $2m$ depending on the parity of $m$. The remaining possibility is that $C$ traverses edges labeled by a combination of $r_1$ or $\flop_1$, and $r_2$ or $\flop_2$. In this case, since the girth of the complete bipartite graph $\up{m}{2}/\up{m}{1}$ is 4, it follows that $|C|\geq5$ (4 edges labeled by $r_2$ or $\flop_2$ and at least one edge labeled by $r_1$ or $\flop_1$.) Therefore, since $(r_1)^m$ is an $m$-cycle, if $m\leq 5$, the girth of $\up{m}{2}$ is $m$. 

Now we consider the case $m > 5$. First, let us notice that $C=\flop_1\flop_2^2r_1r_2^2$ is a 6-cycle in $\up{m}{2}$. Indeed, due to vertex-transitivity, one needs only to see what happens when we apply $C$ to the identity generalized permutation $e=1^02^0$. It can be readily verified that
 \[   (\flop_1\flop_2^2r_1r_2^2)(e) = e.\]

It remains to show that there cannot be a cycle of length five in $\up{m}{2}$, and without loss of generality we assume that the starting vertex is $e$. If there were such a 5-cycle, it has to traverse edges labelled by $r_2$ or $\flop_2$. In other words it must overlap with a cycle within $\up{m}{2}/\up{m}{1}$ by Lemma~\ref{lem:quotient}, which are of length at least 4. Furthermore, since $m > 5$, any 4-cycle in $\up{m}{2}/\up{m}{1}$ must contain two edges labelled $r_2$ and two edges labelled $\flop_2$ to ensure the signs return to their initial value. However, a single $r_1$ or $\flop_1$ within a copy of $\up{m}{1}$ would force the parity of the sign for one symbol to be odd, and thus we could not return to our initial vertex $e$ with five edges since $m>5$. It follows that there is no such 5-cycle, and therefore the girth of $\up{m}{2}$ is $\min\{m,6\}$.
\end{proof}

In the next Theorem, we establish that the girth of $\up{m}{n}$, with $n>2$, is also $\min\{m,6\}$.


\begin{thm}\label{t:girth}
For $m\geq3$ and $n\geq2$, the girth of $\up{m}{n}$ is $\min\{m,6\}$.
\end{thm}
\begin{proof}
Due to the hierarchical structure of $\up{m}{n}$, there are copies of $\up{m}{2}$ embedded in $\up{m}{n}$ and thus the girth of $\up{m}{n}$ is less than or equal to the girth of $\up{m}{2}$, that is $\min\{m,6\}$. We now shall show that there are no shorter cycles present as $n$ increases when $m > 6$. 

We prove this by induction, with $n=2$ from Theorem~\ref{t:n=2girths} being our base case. So let us assume as our inductive hypothesis that the girth of $\up{m}{k-1}$ is $6$ for some $k>3$. Let us assume for the sake of contradiction that the girth of $\up{m}{k}$ is $g \in \{3,4,5\}$, that is $g < 6$, and let $C$ be a cycle of length $g$ in $\up{m}{k}$. Note that $C$ cannot be completely contained in a copy of $\up{m}{k-1}$ embedded in $\up{m}{k}$ because of the inductive hypothesis that the girth of $\up{m}{k-1}$ is $6$. Thus $C$ traverses edges labeled $r_{k}$ of $\flop_{k}$, and let $C'$ be the cycle corresponding to $C$ in $\up{m}{k}/\up{m}{k-1}$ given in the statement of Lemma~\ref{lem:quotient}. 

As mentioned in the paragraph before Lemma~\ref{lem:quotient} we know that the quotient graph $\up{m}{n}/\up{m}{n-1}$ is a complete $n$-partite graph, whose parts are of size $m$, one for each possible sign. It is easy to verify that the girth of this quotient graph is 3. So the cycle $C'$ in $\up{m}{k}/\up{m}{k-1}$ could be of length $3,4,$ or $5$. We shall show that in each case the length of the corresponding cycle $C$ in $\up{m}{k}$ cannot be of length less than 6. Without loss of generality, due to the vertex-transitivity of $\up{m}{k}$, we may consider our cycle $C$ starts at the identity generalized permutation $e=1^02^0 \cdots k^0 \in S(m,k)$.


If $|C'|=3$, then $C'$ connects three different copies of $\up{m}{k-1}$ embedded in $\up{m}{k}$ and $C'$ must be of the form $r_k^3$, $r_k^2\flop_k$, $r_k\flop_k^2$ or $\flop_k^3$. (Since $C'$ is a cycle, forms like $r_k\flop_k r_k$ are included in $r_k^2\flop_k$.)  Clearly $C'\neq C$ since 
\begin{align*}
    (r_k^3)(e) &= k^3(k-1)^3 \cdots 2^31^3, \\  
    (r_k^2\flop_k)(e) &= k^1(k-1)^1 \cdots 2^11^1,\\
    (r_k\flop_k^2)(e) &= k^{m-1}(k-1)^{m-1} \cdots 2^{m-1}1^{m-1}, \text{ and }\\
    (\flop_k^3)(e) &= k^{m-3}(k-1)^{m-3} \cdots 2^{m-3}1^{m-3}    
\end{align*} 
are not the identity. Thus $C$ contains edges within at least one of the three copies of $\up{m}{k-1}$ that $C'$ is incident upon. 

Consider when $C'=r_k^3$ or $C'=r_k^2\flop_k$. The last signed-symbol of the three distinct copies would be $k^0$, $1^1$, and $i^{m_i}$ for some $i \in [2,k-1]$ and $m_i\in[0,m-1]$. We traverse $C'$ from $e=1^02^0 \cdots k^0$ and perform the first $r_k$ in $C'$ to reach the next copy of $\up{m}{k-1}$, which ends in $1^1$,
\begin{equation}\label{e:1r_k} 
(r_k)(e) = k^1 (k-1)^1 \cdots i^1 \cdots 2^11^1. 
\end{equation}

In order to reach the next copy of $\up{m}{k-1}$, which ends in $i^{m_i}$, we can either flip or flop with index $k-i+1$ to put symbol $i$ in the first position. Suppose the reversal done is $\flop_{k-i+1}$ followed by the second $r_k$ of $C'$,
\begin{equation} 
(\flop_{k-i+1}r_k)(e) = i^0 \cdots (k-1)^0k^0(i-1)^1\cdots2^11^1\text{, and}
\end{equation}

\begin{equation}\label{e:2r_k} 
(r_k\flop_{k-i+1}r_k)(e) = 1^2 2^2 \cdots (i-1)^2 k^1 \cdots i^1.\end{equation}

Before completing the last flip/flop of $C'$ to enter our last copy of $\up{m}{k-1}$, which ends in $k^0$, at least two more reversals are required to not only place the symbol $k$ at the beginning of the generalized permutation but also to make its sign within $\pm1$ of $0$. A similar argument would suffice for $C'=\flop_k^3$ or $C'=\flop_k^2r_k$. Thus when $|C'|=3$ the associated cycle $C$ must have length $|C| \geq 6$. The same argument holds if we perform $r_{k-i+1}$ followed by the second $r_k$

If $|C'| = 4$, then $C'$ must be of the form $r_k^4$, $r_k^3\flop_k$, $r_k^2\flop_k^2$, $(r_k\flop_k)^2$, $r_k\flop_k^3$, or $\flop_k^4$. Once again it can be verified that each possible form of $C'$ is not identically $C$ since 
\begin{align*}
    (r_k^4)(e) &= 1^42^4 \cdots (k-1)^4k^4, \\  
    (r_k^3\flop_k)(e) &= 1^22^2 \cdots (k-1)^2k^2,\\
    (r_k\flop_k^3)(e) &= 1^{m-2}2^{m-2} \cdots (k-1)^{m-2}k^{m-2}, \text{ and }\\
    (\flop_k^4)(e) &= 1^{m-4}2^{m-4} \cdots (k-1)^{m-4}k^{m-4}    
\end{align*} 
while $r_k^2\flop_k^2$ and $(r_k\flop_k)^2$ would traverse the same edge forward and backward in $\up{m}{k}$. Thus some additional edges in the copies of $\up{m}{k-1}$ that $C'$ is incident upon are needed so that $C$ is a cycle.

In the cases where $C'$ is of the form $r_k^2\flop_k^2$, $(r_k\flop_k)^2$, $r_k^3\flop_k$, and $r_k\flop_k^3$ at least two edges are needed between a pair of $r_k$ and $\flop_k$ to avoid the traversing of the same edge. There are two pairs of adjacent $r_k$'s and $\flop_k$'s. So any associated cycle $C$ in $\up{m}{k}$ must have $|C| \geq 6$.

Lastly consider the cases where $C'$ is of the form $r_k^4$ and $\flop_k^4$. Say $C'=r_k^4$ and a similar argument suffices for $C'=\flop_k^4$. The last signed-symbol of the four distinct copies met by $C'$ would be $k^0, 1^1, i^{m_i},$ and $j^{m_j}$ for $i\in[2,k]$, $j\in[k-1]\smallsetminus\{i\}$, and $m_i,m_j \in [0,m-1]$. We start traversing from $e$ exiting the copy with last signed-symbol $k^0$ and leading to the copy ending with signed-symbol $1^1$. Within the latter copy a reversal of length $k-i+1$ is needed to put $i$ in the first position. Say the reversal $\flop_{k-i+1}$ is done followed by the second $r_k$ of $C'$ resulting in the same generalized permutation as in (\ref{e:2r_k}). An additional edge with length less than $k$ would be needed to move the symbol $j$ to the first position so that after the third $r_k$ of $C'$ we enter the copy of $\up{m}{k-1}$ ending in $j^{m_j}$. This gives us that $|C| \geq 6$. The same argument holds if we perform $r_{k-i+1}$ followed by the second $r_k$

Finally if $|C'| = 5$, then it can be verified that $C'\neq C$ in $\up{m}{k}$ since any such $5$-cycle must have an odd number of edges labeled either $r_k$ or $\flop_k$, which would result in the sign of all the symbols to be odd, i.e.\ not $0$ since $m>6$. Thus an additional edge is required in at least one copy of $\up{m}{k-1}$. Therefore the associated cycle $C$ would have $|C| \geq 6$.
\end{proof}

Table~\ref{tab:girths} summarizes the girth of the undirected generalized pancake graphs.

\begin{table}[ht!]
\begin{tabular}{|r|r|l|}
\hline
\multicolumn{1}{|l|}{$m$} & \multicolumn{1}{l|}{$n$} & girth of $\upmn$ \\ \hline
    $1$ & $\geq3$ & $6$ \cite[Theorem 5]{Compeau2011}\\ \hline 
    $2$ & $\geq2$ & $8$ \cite[Theorem 10]{Compeau2011}\\ \hline 
    $\geq 3$ & $\geq2$ & $\min\{m,6\}$\\ \hline 
\end{tabular}
\caption{Girths of $\upmn$ for all non-trivial values of $m,n$.}
\label{tab:girths}
\end{table}

\begin{rem}We remark once again that the cycle structure of $\pmn$ is more mysterious. Clearly, the girth of $\pmn$ is at least $m$ due to the action of $C_m$. However, it remains to be understood whether or not shorter cycles exist.  
\end{rem}

In the remainder of the paper, we discuss the pancyclic and panevencyclic properties of $\upmn.$

\section{\texorpdfstring{$\up{3}{n}$}{P3(n)} is pancyclic}\label{sec:up3n}

In this section, we prove the following theorem.

\begin{thm}\label{t:3pancyclic}
The graph $\up{3}{n}$ is pancyclic for $n\geq1$.
\end{thm}
If $n>3$, we use the ``base cycle" $\C$ from Lemma~\ref{lem:basecycle} that goes through all the $3m$ copies of $\up{3}{n-1}$ embedded in $\up{3}{n}$.  
Moreover, if we take a cycle in any of the copies of $\up{3}{n-1}$ with more than $3^{n-2}(n-2)!$ edges, then such a cycle must contain an edge labeled by $r_{n-1}$. By Lemma~\ref{l:edgetrans}, we can assume that such an edge overlaps with one of the edges labeled $r_{n-1}$ from $\C$. So we may merge $q$ cycles $C_1,\ldots,C_q$ from $q$ different copies of $\up{3}{n-1}$ embedded in $\up{3}{n}$ and $\ell(C_i)>3^{n-2}(n-2)!$ for $i\in [q]$ with the base cycle $\C$. The length of the resulting cycle is

\begin{equation}\label{c1}
    \sum_{i=1}^q(\ell(C_i)-1)+(6n-q), \text{ with } q\in [3n].
\end{equation}

\begin{obs}\label{o:6n-q}
If $q \in [3n]$ and $3^{n-2}(n-2)!<\ell(C_i)\leq 3^{n-1}(n-1)!$ for every $i \in [q]$, then 
\[
3^{n-2}(n-2)!+6n-1\leq \sum_{i=1}^q(\ell(C_i)-1)+(6n-q)\leq 3^{n}n!.
\] Moreover, the lower bound is obtained by setting $q=1$ with $\ell(C_1)=3^{n-2}(n-2)!+1$ and the upper bound is obtained by setting $q=3n$ with $\ell(C_1)=\cdots=\ell(C_{3n})=3^{n-1}(n-1)!$. 
\end{obs}

In the proof of Theorem~\ref{t:3pancyclic}, we will assume as inductive hypothesis that all the lengths from $3^{n-2}(n-2)!+1$ to $3^{n-1}(n-1)$ are possible for each cycle $C_1,\ldots, C_q$.

So it is sufficient to prove that as we vary the length of the cycles in Form (\ref{c1}), we obtain the cycles of every length from $3^{n-2}(n-2)!+6n-1$ to $3^nn!$, which we know are possible by Observation~\ref{o:6n-q}. Furthermore, we know through Observation~\ref{o:overlap} that cycles in Form (\ref{c1}) have lengths overlapping those lengths implied by the inductive hypothesis. The following observations are easily verified.
\begin{obs}\label{o:overlap}
If $n\geq4$, then $3^{n-2}(n-2)!+6n-1<3^{n-1}(n-1)!$.
\end{obs}
\begin{obs}\label{o:bigger n-2}
If $n\geq2$, then $3^{n-1}(n-1)!-3^{n-2}(n-2)!>3^{n-2}(n-2)!$.
\end{obs}

\begin{obs}\label{o:+3}
If $n\geq 3$, then $3^{n-2}(n-2)!+3<3^{n-1}(n-1)!$.
\end{obs}

We now can proceed to the main result of this section.

\begin{proof}[Proof of Theorem~\ref{t:3pancyclic}]
The proof is by induction on $n$. For the base cases of $ n\in [3]$, the claim is verified by hand or by computer search. So we will assume in the inductive hypothesis that $n\geq4$ and that $\up{3}{n-1}$ is pancyclic.

We now show that there are cycles of any length $\ell \in [3,3^{n}n!]$ in $\up{3}{n}$. By the inductive hypothesis we know there already exist cycles of length $\ell \in [3,3^{n-1}(n-1)!]$ within a copy of $\up{3}{n-1}$ with fixed signed-symbol. Thus it remains to verify that there are cycles of length $\ell \in [3^{n-1}(n-1)!,3^{n}n!]$ within $\up{3}{n}$.

These cycles are constructed by merging upon the base cycle, $\C$, cycles of varying lengths within the incident copies of $\up{3}{n-1}$. The length of such a cycle will be in the form (\ref{c1}). Notice, however, that the minimum cycle length obtained from (\ref{c1}) is $\min(\ell(\C))=3^{n-2}(n-2)!+6n-1$, which is greater than $3^{n-2}(n-2)!$. We provide a procedure to produce such a cycle of length $k+1$, if there exists such a cycle of length $k < 3^{n}n!$. 

Let us suppose $k= \sum_{i=1}^q(\ell(C_i)-1)+(6n-q)$ for some $q\in [3n]$ and $\ell(C_i)>3^{n-2}(n-2)!$ for $i \in [q]$. Moreover, notice that since $k<3^{n}n!$, it follows that $q<3n$. Now there are two cases to consider.
\begin{description}
\item[$\ell(C_j)<3^{n-1}(n-1)!$ for some $j$] In this case, we can replace $C_j$ with a cycle $C'_j$ of length $\ell(C_j)+1$, which exists since we are assuming that $\up{3}{n-1}$ is pancyclic. Thus by merging $C_1,\ldots,C'_j,\ldots,C_q$ with $\C$, we obtain a cycle of length $k+1$. 

\item[$\ell(C_i)=3^{n-1}(n-1)!$ for all $1 \leq i < q$] In this case, we can replace $C_q$ with a cycle $C'_q$ of length $3^{n-1}(n-1)!-3^{n-2}(n-2)!$, which is greater than $3^{n-2}(n-2)!$ by Observation~\ref{o:bigger n-2}, and in a copy of $\up{3}{n-1}$ different from the $q<3n$ copies where the cycles $C_1,\ldots,C_q$ are embedded, use a cycle $C_{q+1}$ of length $3^{n-2}(n-2)!+3$, which we know is less than $3^{n-1}(n-1)!$ by Observation~\ref{o:+3}. Such a $C_{q+1}$ exists due to the assumption that $\up{3}{n-1}$ is pancyclic. Hence, after merging $\C$ with $C_1,\ldots,C_{q-1},C'_q,C_{q+1}$ along the edges they have in common, the resulting cycle has length 
\begin{align*}
    &\sum_{i=1}^{q-1}(\ell(C_i)-1)+(\ell(C'_q)-1)+(\ell(C_{q+1})-1)+(6n-(q+1))\\
    &=\sum_{i=1}^{q-1}(\ell(C_i)-1)+\left(3^{n-1}(n-1)!\right)+(6n-q)\\
    &=\sum_{i=1}^{q-1}(\ell(C_i)-1)+(3^{n-1}(n-1)!-1)+(6n-q+1)\\
    &=\sum_{i=1}^q(\ell(C_i)-1)+(6n-q)+1\\
    &=k+1,
\end{align*} as desired.
\end{description}

Thus, there are cycles of any length $\ell \in [3,3^{n}n!]$ within $\up{3}{n}$.  Therefore, $\up{3}{n}$ is pancyclic. 
\end{proof}

We remark that when considering only forward flips, that is, in the directed graph $\p{3}{n}$, the digraph is not pancyclic. In particular, there are no cycles of length 5, 7, 11, and 17 (see Table~\ref{tab:directed-cycles}.) In the next section we show that $\up{4}{n}$ has all cycles of even length, i.e. $\up{4}{n}$ is panevencyclic. 

\section{\texorpdfstring{$\up{4}{n}$}{P4(n)} is panevencyclic}\label{sec:up4n}

In this section we prove our second main result, that $\up{4}{n}$ contains all cycles of length $4, 6, 8,\ldots, 4^nn!$ if $n\geq1$. Furthermore, we point out that unlike $\p{3}{n}$, we are able to understand the structure of the directed $\p{4}{n}$. 

The proof is similar to that of Theorem \ref{t:3pancyclic} in that it is done by induction on $n$. We will employ the base cycle $\C$, again, and merge it with cycles from the different copies of $\up{4}{n-1}$ with $n\geq3$ that are embedded within $\up{4}{n}$.

Given the hierarchical structure of $\up{4}{n}$, there are $4n$ copies of $\up{4}{n-1}$, and we shall embed cycles $C_i$, with $i \in [q]$, from $q$ distinct copies, with $q\in [4n]$, and merge them with the base cycle $\C$. The length of the resulting cycle is 
\begin{equation}\label{e:4n}
\sum_{i=1}^q(\ell(C_i)-1)+(8n-q),
\end{equation}
where $8n$ is the length of $\C$. Furthermore, we shall ensure that the cycles $C_i$ have $\ell(C_i)>4^{n-2}(n-2)!$ so that each $C_i$ will contain an edge labeled $r_{n-1}$, and therefore by edge-transitivity (Lemma~\ref{l:edgetrans}), we pick an edge in $C_i$ that overlaps with $\C$.

The following observations are in order. 

\begin{obs}\label{obs:4n}
If $q\in [4n]$ and $4^{n-2}(n-2)!<\ell(C_i)\leq 4^{n-1}(n-1)!$ for $i\in [q]$, then
\[4^{n-2}(n-2)!+8n-1\leq \sum_{i=1}^q(\ell(C_i)-1)+(8n-q)\leq 4^nn!\]

The lower bound is obtained by setting $q=1$ and $\ell(C_1)=4^{n-2}(n-2)!+1$ and the upper bound is obtained by setting $q=4n$ and $\ell(C_1)=\ldots=\ell(C_{4n})=4^{n-1}(n-1)!.$
\end{obs}

\begin{obs}\label{obs:min4n}
If $n\geq3$, then $4^{n-2}(n-2)!+8n-1<4^{n-1}(n-1)!$.
\end{obs}
\begin{obs}\label{o4:bigger n-2}
If $n\geq2$, then $4^{n-1}(n-1)!-4^{n-2}(n-2)!>4^{n-2}(n-2)!$.
\end{obs}

\begin{obs}\label{o4:+3}
If $n\geq3$, then  $4^{n-2}(n-2)!+4<4^{n-1}(n-1)!$.
\end{obs}

With these established results we can state and prove the following result.

\begin{thm}\label{t:up4n} The graph $\up{4}{n}$ is panevencyclic for $n\geq1$.
\end{thm}
\begin{proof}
The proof is by induction on $n$. If $n=1$, $\p{4}{n}$ is a 4-cycle. If $n=2$, we used a computer search and found the existence of all even cycles with lengths in $4,6,\ldots,32$. (see Table~\ref{tab:undirected-cycles}.) Thus we will assume in the inductive hypothesis that $n\geq3$ and that $\up{4}{n-1}$ has all even cycles from $4,\ldots, 4^{n-1}(n-1)!$. 

We now show that there are even cycles of length $\ell \in \{4,6, \ldots, 4^{n}n! \}$. By the inductive hypothesis we know that there exist cycles of length $\ell \in \{4,6, \ldots, 4^{n-1}(n-1)!\}$ within a single copy of $\up{4}{n-1}$ embedded in $\up{4}{n}$. Thus it remains to verify that there are even cycles of length $\ell \in \{4^{n-1}(n-1)!+2, \ldots, 4^{n}n!\}$ within $\up{4}{n}$.

These cycles are constructed by merging upon the base cycle, $\C$, cycles of particular length within the incident copies of $\up{4}{n-1}$. By Observation~\ref{obs:4n}~\&~\ref{obs:min4n}, the length of a cycle formed by merging cycles on the base cycle, in the form (\ref{e:4n}), overlaps with the lengths assumed within a single copy of $\up{4}{n-1}$. Thus, there is no gap in lengths between those available by the inductive hypothesis and those that can be found by merging with the base cycle. Therefore, since the induction hypothesis gives that each copy of $\up{4}{n-1}$ has cycles of all even lengths from 4 to $4^{n-1}(n-1)!$, we need only show that all remaining even lengths may be found by merging cycles to the base cycle. We shall show that if a cycle of length $k$ can be found by this process, that a cycle of length $k+2$ may be found, up until $k+2 = 4^{n}n!$.

Let us assume that $k= \sum_{i=1}^q(\ell(C_i)-1)+(8n-q)$ for some $q\in [4n]$ and $\ell(C_i)>4^{n-2}(n-2)!$ for $i\in [q]$. Moreover, notice that since $k<4^{n}n!$, it follows that $q<4n$. Now there are two cases to consider.
\begin{description}
\item[$\ell(C_j)<4^{n-1}(n-1)!$ for some $j$] In this case, we can replace $C_j$ with a cycle $C'_j$ of length $\ell(C_j)+2$, which exists by assumption. Thus, by merging $C_1,\ldots,C'_j,\ldots,C_q$ with $\C$, we obtain a cycle of length $k+2$. 
\item[$\ell(C_i)=4^{n-1}(n-1)!$ for all $1\leq i< q$]. In this case, we can replace $C_q$ with a cycle $C'_q$ of length $4^{n-1}(n-1)!-4^{n-2}(n-2)!$, which is greater than $4^{n-2}(n-2)!$ by Observation~\ref{o4:bigger n-2} and in a copy of $\up{4}{n-1}$ different from the $q<4n$ copies where the $C_1,\ldots,C_q$ are embedded, use a cycle $C_{q+1}$ of length $4^{n-2}(n-2)!+4$, which is less than $4^{n-1}(n-1)!$ by Observation~\ref{o4:+3}. Notice that $C_{q+1}$ exists since we are assuming that $\up{4}{n-1}$ contains any even cycle length in from $4$ to $4^{n-1}(n-1)!$. Hence, after merging $\C$ with $C_1,\ldots,C'_q,C_{q+1}$ along the edges they have in common, the resulting cycle has length 
\begin{align*}
    &\sum_{i=1}^{q-1}(\ell(C_i)-1)+(\ell(C'_q)-1)+(\ell(C_{q+1})-1)+(8n-(q+1))\\
    &=\sum_{i=1}^{q-1}(\ell(C_i)-1)+4^{n-1}(n-1)!+(8n-q)+1\\
    &=\sum_{i=1}^{q-1}\ell(C_i)+(4^{n-1}(n-1)!-1)+(8n-q+2)\\
    &=\sum_{i=1}^q(\ell(C_i)-1)+(8n-q)+2\\
    &=k+2,
\end{align*}
\end{description} 
as desired.
\end{proof}

In the next section, we generalize Theorem~\ref{t:3pancyclic} and Theorem~\ref{t:up4n}. The proof is by double induction, with these theorems serving as base case.

\section{\texorpdfstring{$\up{m}{n}$}{Pm(n)} is \texorpdfstring{$m'$}{m'}-pancyclic or \texorpdfstring{$m'$}{m'}-panevencyclic, \texorpdfstring{$m' = \min\{m,6\}$}{m'=min(m,6)} }\label{sec:generalcase}

In this section we prove the general case of Theorem~\ref{t:3pancyclic} and Theorem~\ref{t:up4n}, depending on the parity of $m$. Throughout we assume that $m \geq 3$ and $n \geq 1$. Note, in the case when $n=1$ the graph $\up{m}{1}$ is an $m$-cycle and trivially $m$-pancyclic or $m$-panevencyclic. Furthermore, we shall use the notation $m'$ for $\min\{m,6\}$, the girth of $\up{m}{n}$. The computer searches, for small $m$ and $n$, on $\upmn$ indicated an increase in girth, tracking with the number of signs $m$, but still having embedded cycles with lengths in $[m, m^nn!]$, for odd $m$, or lengths in $\{2k \mid k \in [m/2,m^nn!/2] \}$, for even $m$. With those verified lower values, $(m,n)=(3,2)$ and $(m,n)=(4,2)$, as base cases and Theorems \ref{t:3pancyclic} and \ref{t:up4n} as one direction of a double induction proof, for each parity of $m$, we shall complete the other remaining directions herein. A double induction proof can be done to prove a claim $C(m,n)$ for $m \geq a$ and $n \geq b$, for integers $a$ and $b$ by verifying  
    \begin{enumerate}
        \item $C(a,b)$ is true,
        \item by induction that $C(a,n)$ is true for all $n \geq b$, and
        \item by induction that $C(m,n)$ is true for fixed $n$ and all $m \geq a$, 
    \end{enumerate}
see for reference~\cite[Chapter 3]{Gunderson}.   
That is, we shall show that assuming $\upmn$ is either $m$-pancyclic or $m$-panevencyclic for a fixed $n$, then so is $\up{m+2}{n}$ is $(m+2)$-pancyclic or $(m+2)$-panevencyclic when $m$ is odd or even, respectively. Prior to that, though, we point out some prescient observations.

We will once again employ the same strategy of merging the ``base cycle'' $\C$ with cycles $C_1, \ldots, C_q$ in $q\in[mn]$ distinct copies of $\up{m}{n-1}$ within $\up{m}{n}$. The length of the resulting cycle is

\begin{equation}\label{c1mn}
    \sum_{i=1}^q(\ell(C_i)-1)+(2mn-q), \text{ with } q\in [mn].
\end{equation}

\begin{obs}\label{o:(m,n)mergecycles}
If $q\in[mn]$ and $m^{n-2}(n-2)! < \ell(C_i) \leq m^{n-1}(n-1)!$ for every $i\in[q]$, then 
\[
    m^{n-2}(n-2)!+2mn-1 \leq \sum_{i=1}^q(\ell(C_i)-1)+(2mn-q)\leq m^{n}n!.
\] Moreover, the lower bound is obtained by setting $q=1$ with $\ell(C_1)=m^{n-2}(n-2)!+1$ and the upper bound is obtained by setting $q=mn$ with $\ell(C_1)=\cdots=\ell(C_{mn})=m^{n-1}(n-1)!$. 
\end{obs}

\begin{obs}\label{obs:minmn}
If $m\geq3$ and $n\geq3$, then $m^{n-2}(n-2)!+2mn-1<m^{n-1}(n-1)!$.
\end{obs}
\begin{obs}\label{om:bigger n-2}
If $m\geq3$ and $n\geq2$, then $m^{n-1}(n-1)!-m^{n-2}(n-2)!>m^{n-2}(n-2)!$.
\end{obs}
\begin{obs}\label{om:+3or+4}
If $m\geq3$ and $n\geq3$, then $m^{n-2}(n-2)!+3<m^{n-1}(n-1)!$ and $m^{n-2}(n-2)!+4<m^{n-1}(n-1)!$.
\end{obs}

Now we state and prove the general case including both parities of $m$.

\begin{thm}\label{t:m-pancyclic}
For $m\geq3$ and $n\geq1$, the graph $\up{m}{n}$ is
    \begin{enumerate}[(i)]
        \item $m'$-pancyclic, for odd $m$, or
        \item $m'$-panevencyclic, for even $m$,
    \end{enumerate}
    where $m' = \min\{m,6\}$.
\end{thm}

\begin{proof}
The proof shall be done by a double induction, upon $m$ and $n$. The cases where $n=1$ are trivially true since $\up{m}{1}$ is itself an $m$-cycle. The non-trivial base cases would be $\up{3}{2}$ and $\up{4}{2}$ for the even and odd cases for $m$. Both have been verified by computer search, see Table \ref{tab:undirected-cycles}. Furthermore, we have already shown the induction on the $n$ direction when $m=3$ in Theorem \ref{t:3pancyclic} and the induction on the $n$ direction when $m=4$ in Theorem \ref{t:up4n}. Thus, for both parities of $m$ we must show the results hold for a fixed $n$ and as we increase $m$ by two. Fix $n \geq 3$, and let $m' = \min\{m,6\}$.

\begin{description}
\item[$m$ is odd] 
We now show that there are cycles of any length $\ell \in [m',m^{n}n!]$ in $\up{m}{n}$. By the inductive hypothesis we know there already exist cycles of length $\ell \in [m',m^{n-1}(n-1)!]$ within a copy of $\up{m}{n-1}$ with fixed signed-symbol. Thus it remains to verify that there are cycles of length $\ell \in [m^{n-1}(n-1)!,m^{n}n!]$ within $\up{m}{n}$.

These cycles are constructed by merging upon the base cycle, $\C$, cycles of varying lengths within the incident copies of $\up{m}{n-1}$. The length of such a cycle will be in the form (\ref{c1mn}). Notice, however, that the minimum cycle length obtained from (\ref{c1mn}) is $\min(\ell(\C))=m^{n-2}(n-2)!+2mn-1$, which is greater than $m^{n-2}(n-2)!$ and $m'$. Also, by Observation~\ref{obs:minmn}, the lengths of form (\ref{c1mn}) overlap with the lengths of cycles that exist by the inductive hypothesis and can have greater lengths by Observation~\ref{o:(m,n)mergecycles}. We provide a procedure to produce such a cycle of length $k+1$, if there exists such a cycle of length $k < m^{n}n!$. 

Let us suppose $k= \sum_{i=1}^q(\ell(C_i)-1)+(2mn-q)$ for some $q\in [mn]$ and $\ell(C_i)>m^{n-2}(n-2)!$ for $i \in [q]$. Moreover, notice that since $k<m^{n}n!$, it follows that $q<mn$. Now there are two cases to consider.
\begin{description}
\item[$\ell(C_j)<m^{n-1}(n-1)!$ for some $j$] In this case, we can replace $C_j$ with a cycle $C'_j$ of length $\ell(C_j)+1$, which exists since we are assuming that $\up{m}{n-1}$ is $m'$-pancyclic. Thus by merging $C_1,\ldots,C'_j,\ldots,C_q$ with $\C$, we obtain a cycle of length $k+1$. 

\item[$\ell(C_i)=m^{n-1}(n-1)!$ for all $1 \leq i < q$] In this case, we can replace $C_q$ with a cycle $C'_q$ of length $m^{n-1}(n-1)!-m^{n-2}(n-2)!$, which is greater than $m^{n-2}(n-2)!$ by Observation~\ref{om:bigger n-2}, and in a copy of $\up{m}{n-1}$ different from the $q<mn$ copies where the cycles $C_1,\ldots,C_q$ are embedded, use a cycle $C_{q+1}$ of length $m^{n-2}(n-2)!+3$, which we know is less than $m^{n-1}(n-1)!$ by Observation~\ref{om:+3or+4}. Such a $C_{q+1}$ exists due to the assumption that $\up{m}{n-1}$ is $m$-pancyclic. Hence, after merging $\C$ with $C_1,\ldots,C_{q-1},C'_q,C_{q+1}$ along the edges they have in common, the resulting cycle has length 
\begin{align*}
    &\sum_{i=1}^{q-1}(\ell(C_i)-1)+(\ell(C'_q)-1)+(\ell(C_{q+1})-1)+(2mn-(q+1))\\
    &=\sum_{i=1}^{q-1}(\ell(C_i)-1)+\left(m^{n-1}(n-1)!\right)+(2mn-q)\\
    &=\sum_{i=1}^{q-1}(\ell(C_i)-1)+(m^{n-1}(n-1)!-1)+(2mn-q+1)\\
    &=\sum_{i=1}^q(\ell(C_i)-1)+(2mn-q)+1\\
    &=k+1,
\end{align*} as desired.
\end{description}

Thus, there are cycles of any length $\ell \in [m',m^{n}n!]$ within $\up{m}{n}$.  Therefore, $\up{m}{n}$ is $m'$-pancyclic.

\item[$m$ is even] 

The case for even $m$ is nearly identical. However, the fact that the base case of the induction being panevencyclic forces the inductive hypothesis to assume there are even length cycles in the copies of $\up{m}{n-1}$ within $\up{m}{n}$. Furthermore, the length of the cycle $C_{q+1}$ would need to be $m^{n-2}(n-2)!+4$, not $m^{n-2}(n-2)!+3$, as it was in the case when $m$ was odd. This is to ensure the cycle is of even length.

Thus, there are cycles of any even length $\ell \in \{2k \mid k \in [m'/2,m^{n}n!/2]\}$ within $\up{m}{n}$.  Therefore, $\up{m}{n}$ is $m'$-panevencyclic.
\end{description} 

This concludes the proof. 
\end{proof}

In our study, we used computers to find lengths of cycles in $\pmn$ and $\upmn$. This proved to be a bit challenging due to the number of vertices of these graphs is $m^nn!$. We modify an existing algorithm~\cite{HJ08} and implement it in parallel. We include the details in the following section. 

\section{Computer search for lengths of cycles in \texorpdfstring{$\pmn$}{P(m,n)} and \texorpdfstring{$\upmn$}{Pm(n)}.}\label{sec:computer}

It would be interesting to know the cycles that can be embedded in $\pmn$ for different values of $m,n$. For this, we implement a modified version of the algorithm presented in~\cite{HJ08} to find cycles of different lengths. The first modification is that our algorithm is running on multiple threads, with the problem of finding a cycle of a given length divided by the first initial moves. Up to a given depth, a new thread is spawned for every neighbor of a vertex $v$ seen in the search. From the thread's initial given moves forward, each thread searches depth-first for cycles of lengths that have not yet been seen. If a given graph has $n$ vertices, once a cycle of length $n$ is found, every thread becomes limited in what length it will search up to, starting from $n-1$, so as to not waste computation power looking for new cycles of large sizes that we have already found. Notably, this search space restriction is strictly from the top-down, and if the restriction is occurring for cycles of length $k$, then every cycle length from $k$ up to, and including $n$, has been found. Because we cannot rule out whether a cycle exists for a given length, without exhaustively searching every possible path smaller than it, this optimization is only beneficial for graphs with cycles of larger length. Otherwise, we cannot limit the search space, as we will have exhausted the search space entirely by the time we are able to conclusively shrink it, if large cycle lengths do not exist.

Additionally, our algorithm only explores cycles from one given vertex since $\pmn$ is vertex-transitive. From there, the algorithm saves to a file the exact cycle found every time it encounters a cycle of a length that it has not yet seen, and stores that cycle length in shared memory between the threads, so the other threads can determine if that cycle length can be used to shrink the search space to smaller cycle lengths. In spite of all these optimizations, we were only able to find the cycle for the values of $m,n$ listed in Table~\ref{tab:directed-cycles} and Table~\ref{tab:undirected-cycles}. The code and instructions can be found in the following repository~\cite{A22}.

\section*{Acknowledgments}

We would like to thank Aaron Leslie and Jordan Graves for their contribution in implementing the algorithm described in Section~\ref{sec:computer}.

\end{document}